\documentclass[12pt,twoside]{amsart}
\usepackage{amsmath,amsthm,amsmath,amsrefs}

\usepackage{calc}
\usepackage{setspace}
\usepackage{amsbsy,amsmath,amsfonts,amsthm,amssymb,color}
\usepackage[margin=1in]{geometry}
\numberwithin{equation}{section}
\theoremstyle{plain}
\newtheorem{theorem}{Theorem}[section] %[section]
\newtheorem{lemma}[theorem]{Lemma}
\newtheorem{definition}[theorem]{Definition}
\newtheorem{corollary}[theorem]{Corollary}
\newtheorem{remark}[theorem]{Remark}
\newtheorem{proposition}[theorem]{Proposition}

\theoremstyle{definition}

\usepackage{mathrsfs}

\newcommand{\bN}{\mathbb{N}}

\newcommand{\cF}{\mathcal{F}}

\newcommand{\cB}{\mathcal{B}}

\newcommand{\cA}{\mathcal{A}}

\newcommand{\cM}{\mathcal{M}}
\newcommand{\cU}{\mathcal{U}}
\newcommand{\bC}{\mathbb{C}}
\newcommand{\bR}{\mathbb{R}}
\newcommand{\cG}{\mathcal{G}}
\newcommand{\cR}{\mathcal{R}}
\newcommand{\id}{\text{id}}

\newcommand{\cI}{\mathcal{I}}
\newcommand{\cJ}{\mathcal{J}}

\usepackage{tikz}
\usepackage{tikz-cd}

\title[Synchronous games with $*$-isomorphic game algebras]{Synchronous games with $*$-isomorphic game algebras}

\author{Samuel J. Harris}
\address{Texas A\&M University\\
Department of Mathematics \\
College Station, TX \\
77843-3368 USA}
\email{sharris@tamu.edu}

\begin{document}
\maketitle

\begin{abstract}
We establish several strong equivalences of synchronous non-local games, in the sense that the corresponding game algebras are $*$-isomorphic. We first show that the game algebra of any synchronous game on $n$ inputs and $k$ outputs is $*$-isomorphic to the game algebra of an associated bisynchronous game on $nk$ inputs and $nk$ outputs. As a result, we show that there are bisynchronous games with equal question and answer sets, whose optimal strategies only exist in the quantum commuting model, and not in the quantum approximate model. Moreover, we exhibit a bisynchronous game with $20$ questions and $20$ answers that has a non-zero game algebra, but no winning commuting strategy, resolving a problem of V.I. Paulsen and M. Rahaman. We also exhibit a $*$-isomorphism between any synchronous game algebra with $n$ questions and $k>3$ answers and a synchronous game algebra with $n(k-2)$ questions and $3$ answers.
\end{abstract}

\maketitle

\section{Introduction}

The theory of two-player non-local games has seen a significant amount of attention in the last number of years. These games, which feature two players that work cooperatively to win a single-round game, can be used as a way to check whether or not the players possess a certain form of quantum entanglement. The two players, usually named Alice and Bob, are each asked a question $x$ and $y$ respectively from a finite question set $I$ by a referee. The players respond with answers $a$ and $b$, respectively, from a finite answer set $O$. The rules of the game are encoded in a rule function $\lambda:O \times O \times I \times I \to \{0,1\}$. The referee declares that the players win the round if their $4$-tuple $(a,b,x,y)$ is allowed; that is, if $\lambda(a,b,x,y)=1$. The players lose if $\lambda(a,b,x,y)=0$. The players are not allowed to communicate once the game begins, but are allowed to agree on a strategy beforehand.

These non-local games directly relate to the sets of probability distributions that are possible in a finite-input, finite-output system, where Alice and Bob are in a separated system (in particular, Alice's measurement operators commute with Bob's measurement operators). Generally, the probability of the players obtaining outputs $a$ and $b$, respectively, given inputs (or experiments) $x$ and $y$, is given by the joint probability $p(a,b|x,y)$. In this way, in a setting where there are $n$ experiments, each with $k$ outcomes, the set of probability distributions form a convex subset of $\mathbb{R}^{n^2k^2}$. The set of all distributions obtained from classical random variables is the set of \textbf{local correlations}, denoted $C_{loc}(n,k)$. The set of all such joint distributions obtained by possessing an entangled state in a finite-dimensional resource Hilbert space is the set of \textbf{quantum correlations}, denoted by $C_q(n,k)$. The set of \textbf{approximately finite-dimensional quantum correlations} is defined as $C_{qa}(n,k)=\overline{C_q(n,k)}$. Lastly, the set of \textbf{quantum commuting correlations} $C_{qc}(n,k)$ is the set of all probability distributions obtained when the players possess a shared entangled state, but possibly over an infinite-dimensional resource Hilbert space. Naturally, elements of the sets $C_t(n,k)$, $t \in \{loc,q,qa,qc\}$ can be thought of as strategies for a non-local game with $n$ questions and $k$ answers. We will call an element of $C_t(n,k)$ a \textbf{winning $t$-strategy} for a game if the players win each round with probability $1$; that is, if $p(a,b|x,y)=0$ whenever $\lambda(a,b,x,y)=0$.

In the last few years, several significant open problems dating back to B.S. Tsirelson \cite{Ts93} have been solved regarding these correlation sets. The strong Tsirelson problem asks whether $C_q(n,k)=C_{qc}(n,k)$ for all $n,k \geq 2$. This was resolved negatively in exciting work of W. Slofstra \cite{Slo19}, who showed that $C_q(n,k) \neq C_{qa}(n,k)$ for large values of $n$ and $k$. The weak Tsirelson problem asks whether $C_{qa}(n,k)=C_{qc}(n,k)$ for all $n,k \geq 2$. This remained open, until very recently a counterexample was exhibited in monumental work of Z. Ji, A. Natarajan, T. Vidick, J. Wright and H. Yuen \cite{JNVWY20}. The weak Tsirelson problem has deep connections in the theory of operator algebras. In particular, an equivalent form is Connes' embedding problem \cite{Co76}, which asks whether every weakly separable finite von Neumann algebra $\cM$ can be approximately embedded into the hyperfinite $II_1$ factor $\mathcal{R}$ in a trace-preserving manner \cite{J+,Fr11,Oz13}.

A special class of non-local games, called synchronous games, has been studied heavily since the introduction of the idea in \cite{PSSTW16}. A game is \textbf{synchronous} if, whenever the players receive identical inputs, they must respond with identical outputs. This innocent condition forces players that win the game with probability $1$ to be in possession of a (quantum) shared function, and the resulting probability distributions relate directly to traces on the $C^*$-algebra generated by Alice's (respectively, Bob's) operators. Depending on the model $t \in \{loc,q,qa\}$, the $C^*$-algebra can be arranged to be abelian, finite-dimensional, or a free ultrapower of the hyperfinite $II_1$ factor $\cR$, respectively. Key to the study of synchronous games is the game algebra $\cA(\cG)$ of a synchronous game, which encodes all winning strategies into a single, universal unital $*$-algebra. These algebras sometimes have strange properties, as outlined in \cite{HMPS19}: there exist synchronous games that have no winning strategies in the $qc$ model, but whose game algebra is non-trivial as a unital $*$-algebra.

A related class of games that has received more attention recently is the class of bisynchronous games, which are synchronous games with the added property that players who receive distinct questions \textit{cannot} respond with identical answers. Regarding these games, two problems were posed by V.I. Paulsen and M. Rahaman in \cite{PR21}:

\begin{itemize}
\item[(1)] If $\cG=(I,O,\lambda)$ is a bisynchronous game with $|I|=|O|$ and $\cA(\cG) \neq \{0\}$, then does $\cG$ have a winning $qc$-strategy?
\item[(2)]
Does there exist a bisynchronous game $\cG=(I,O,\lambda)$ with $|I|=|O|$ such that $\cG$ has a winning $qc$-strategy, but no winning $qa$-strategy? (More generally, is the set of bisynchronous $qa$-correlations in $m$ inputs and $m$ outputs ever distinct from the set of bisynchronous $qc$-correlations in $m$ inputs and $m$ outputs?)
\end{itemize}
In contrast, if one allows for $|I| \neq |O|$, then the answers to both (1) and (2) are known. Indeed, this is since every synchronous game $\cG=(I,O,\lambda)$ with $|I|=n$ and $|O|=k$ is equivalent, in a sense, to a bisynchronous game with $n$ inputs and $nk$ outputs \cite{PR21}, and there are examples of a synchronous game violating (1) (see \cite{HMPS19}) and a synchronous game satisfying (2) (see \cite{JNVWY20}), both without the requirement that $|I|=|O|$. On the other hand, a key example of a bisynchronous game with $|I|=|O|$ is the graph isomorphism game $\text{Iso}(G,H)$ between two finite, simple, undirected graphs $G$ and $H$ with the same number of vertices. For that game, the answer to question (1) is ``yes" \cite{BCEHPSW20}, and the answer to question (2) is still open.

In this paper, we resolve both problems. In particular, we exhibit a bisynchronous game with $20$ questions and $20$ answers that has a non-zero game algebra, but no winning $qc$-strategies, which yields a negative answer to problem (1). We also prove that there is a bisynchronous game with equal question and answer sets that has a winning $qc$-strategy, but no winning $qa$-strategy, which yields a positive answer to problem (2). Our approach uses certain $*$-isomorphisms between synchronous game algebras. Any two synchronous games that have $*$-isomorphic game algebras are very closely related in terms of winning strategies in the different models. In particular, with a sufficiently nice isomorphism of game algebras, one yields affine homeomorphisms of the sets of winning strategies in the different models (see Theorem \ref{theorem: isomorphism gives affine homeomorphisms}). In this way, we prove that, in the models $t \in \{loc,q,qa,qc\}$, everything about winning strategies for synchronous games can be reduced to the bisynchronous setting with equal question and answer sets. We also prove a similar reduction from synchronous games with $k>3$ outputs to those with $3$ outputs, extending a result of T. Fritz \cite{Fr20}.

The paper is organized as follows. In Section \S2, we recall some preliminaries regarding synchronous correlations and synchronous games. We also prove that a certain type of $*$-isomorphism $\cA(\cG_1) \simeq \cA(\cG_2)$ between synchronous game algebras yields affine homeomorphisms between the sets of winning $t$-strategies for the two games, for $t \in \{loc,q,qa,qc\}$. In Section \S3, we prove that every synchronous game $\cG$ on $n$ inputs and $k$ outputs has game algebra isomorphic to the game algebra of a bisynchronous game $\widetilde{\cG}$ with $nk$ inputs and $nk$ outputs (Theorem \ref{theorem: bisynchronous equivalence}). In Section \S4, we obtain a similar isomorphism between synchronous games with $n$ inputs and $k>3$ outputs and synchronous games with $n(k-2)$ inputs and $3$ outputs (Theorem \ref{theorem: equivalence to 3 output game}), yielding an alternate proof that synchronous games can be reduced, in a certain sense, to the three output setting--this fact was originally established in \cite{Fr20}. One corollary of this fact is that there is a synchronous game on $10$ inputs and $3$ outputs that has a non-zero game algebra, but no winning $qc$-strategy.

\section{Preliminaries}

In this section, we collect some facts about synchronous non-local games that we will need, along with a condition on isomorphic game algebras that yields affinely homeomorphic sets of winning strategies in the models $\{loc,q,qa,qc\}$.

A \textbf{non-local game} is a triple $\cG=(I,O,\lambda)$, where $I$ and $O$ are finite sets and $\lambda:O \times O \times I \times I \to \{0,1\}$ is a function called the \textbf{rule function}. We call the game \textbf{synchronous} if $\lambda(a,b,x,x)=0$ whenever $a \neq b$ in $O$. We typically write $n$ for the size of $I$ and $k$ for the size of $O$.

These games are played by two players, who aim to cooperatively win a round of the game by providing correct answers to a referee. Essentially, the referee asks question $x$ to the first player, Alice, and question $y$ to the second player, Bob. Alice and Bob must respond with answers $a$ and $b$, respectively. The players win the round if $\lambda(a,b,x,y)=1$, and they lose if $\lambda(a,b,x,y)=0$. While the players can agree on a strategy before the game begins, they are not allowed to communicate once it starts.

There are several natural sets of probability distributions for Alice and Bob's answers, given the questions they receive. In a classical (or local) model without the use of entanglement, these are the \textbf{local correlations}; in a finite-dimensional tensor product model for entanglement, the strategies are the \textbf{quantum correlations}; in an approximate finite-dimensional model they are the \textbf{quantum approximate correlations}; and in the commuting operator framework they are the \textbf{quantum commuting correlations}.

Thanks to the work of \cite{PSSTW16} and \cite{KPS18}, winning strategies for synchronous games are very nicely described in terms of traces the unital $C^*$-algebra generated by Alice's projections (respectively, Bob's projections); this is the approach that we will take here. One can view this link as a correspondence between synchronous correlations and the players having a quantum version of a shared function.

The set $C_{qc}^s(n,k)$ of all synchronous quantum commuting correlations is the set of all distributions of the form
\[ p(a,b|x,y)=\tau(E_{a,x}E_{b,y}),\]
where $\tau$ is a tracial state on a unital $C^*$-algebra $\cA$ (i.e., $\tau$ is a positive linear functional with $\tau(1)=1$ and $\tau(XY)=\tau(YX)$ for all $X,Y \in \cA$), and $\{E_{a,x}\}_{a=1}^k$ is a projection-valued measure (PVM) in $\cA$ for each $1 \leq x \leq n$ \cite{PSSTW16}. A standard trick involving the GNS representation of the trace allows one to assume without loss of generality that $\tau$ is faithful; that is, $\tau(A^*A)>0$ for any $A \in \cA \setminus \{0\}$ \cite{PSSTW16}.

The description of the set $C_{loc}^s(n,k)$ of synchronous local correlations is analogous to $C_{qc}^s(n,k)$, except that the projections $E_{a,x}$ are assumed to generate an abelian $C^*$-algebra. (In this case, the ``tracial" requirement on the state $\tau$ is redundant.) The set $C_q^s(n,k)$ of synchronous quantum correlations is defined like $C_{qc}^s(n,k)$, with the extra requirement that the algebra $\cA$ be finite-dimensional. Lastly, the set $C_{qa}^s(n,k)$ of all synchronous quantum approximate correlations, is the set of all such correlations where one can arrange for the projections to come from an ultrapower $\cR^{\cU}$ of the hyperfinite $II_1$-factor, using the canonical trace \cite{KPS18}. This set was also shown to be the closure of $C_q^s(n,k)$ in \cite{KPS18}.

It is well-known that
\[ C_{loc}^s(n,k) \subseteq C_q^s(n,k) \subseteq C_{qa}^s(n,k) \subseteq C_{qc}^s(n,k),\]
and all of these sets are convex in $\bR^{n^2k^2}$. Each set, except $C_q^s(n,k)$, is closed in general. Moreover, for large enough $n,k$, all these inclusions are strict \cite{AMRSSV19,KPS18,JNVWY20}.

It is helpful to consider the collection of all winning $t$-strategies for a synchronous game $\cG=(I,O,\lambda)$. where $t \in \{loc,q,qa,qc\}$. If $|I|=n$ and $|O|=k$, then we will define
\[ C_t(\lambda)=\{ (p(a,b|x,y)) \in C_t(n,k): p(a,b|x,y)=0 \text{ for all } (a,b,x,y) \in \lambda^{-1}(\{0\})\}.\]
As $\cG$ is synchronous, it is immediate that $C_t(\lambda) \subset C_t^s(n,k)$.

Key to our work is the universal unital $*$-algebra of a synchronous game $\cG=(I,O,\lambda)$. This algebra, called the \textbf{game algebra}, is the universal unital $*$-algebra $\cA(\cG)$ generated by elements $e_{a,x}$, for $x \in I$ and $a \in O$, satisfying the following conditions:
\begin{enumerate}
\item
$e_{a,x}^2=e_{a,x}=e_{a,x}^*$ for all $x \in I$ and $a \in O$;
\item
$\displaystyle \sum_{a \in O} e_{a,x}=1$ for every $x \in I$; and
\item
$e_{a,x}e_{b,y}=0$ whenever $\lambda(a,b,x,y)=0$.
\end{enumerate}
Alternatively, one forms this algebra by taking the free unital $*$-algebra $\cF(O,I)$ with generators $\widetilde{e}_{a,x}$, $a \in O, \, x \in I$, and taking a quotient of $\cF(O,I)$ by the two-sided $*$-closed ideal $\cI_{\lambda}$ generated by all terms of the form
\begin{align*}
\widetilde{e}_{a,x}^2-\widetilde{e}_{a,x}, \,\,\,\,\,&(a,x) \in O \times I \\
\widetilde{e}_{a,x}-\widetilde{e}_{a,x}^*, \, \,\,\,\,&(a,x) \in O \times I \\
1-\sum_{a \in O} \widetilde{e}_{a,x}, \, \,\,\,\,&x \in I, \\
\widetilde{e}_{a,x}\widetilde{e}_{b,y}, \,\,\,\,\, &(a,b,x,y) \in \lambda^{-1}(\{0\}).
\end{align*}
Depending on the game, this algebra may be trivial in the sense that we may have $\cI_{\lambda}=\cF(O,I)$. In this case, $\cA(\cG)=\{0\}$; i.e., $1=0$ in the game algebra. We say that the game algebra $\cA(\cG)$ \textbf{exists} if $1 \neq 0$ in $\cA(\cG)$.

Game algebras are important because they encode existence of winning strategies for synchronous games. Indeed, the following are true:
\begin{enumerate}
\item
$\cG$ has a winning $loc$-strategy $\iff$ $\cA(\cG)$ has a unital $*$-homomorphism into an abelian unital $C^*$-algebra;
\item
$\cG$ has a winning $q$-strategy $\iff$ $\cA(\cG)$ has a unital $*$-homomorphism into a finite-dimensional $C^*$-algebra;
\item
$\cG$ has a winning $qa$-strategy $\iff$ $\cA(\cG)$ has a unital $*$-homomorphism into $\cR^{\cU}$ for some free ultrapower $\cU$ over $\bN$;
\item
$\cG$ has a winning $qc$-strategy $\iff$ $\cA(\cG)$ has a unital $*$-homomorphism into a unital, tracial $C^*$-algebra $(\cA,\tau)$.
\end{enumerate}

(See \cite{HMPS19} for proofs of (1), (2) and (4). For a proof of (3), see \cite{KPS18}.)

Motivated by these facts, there is also the notion of a \textbf{winning $C^*$-strategy}, which is when $\cA(\cG)$ has a unital $*$-homomorphism into a unital $C^*$-algebra (possibly without a trace) \cite{HMPS19}. The game $\cG$ is said to have a \textbf{winning hereditary strategy} if $\cA(\cG)$ has a unital $*$-homomorphism into a hereditary unital $*$-algebra $\cB$ (that is, one where $\sum_{i=1}^m x_i^*x_i=0$ implies $x_i=0$ for any $x_i \in \cB$ and $m \in \bN$). Lastly, we say that $\cG$ has a \textbf{winning $A^*$-strategy} if $\cA(\cG)$ exists.

If $\cG$ is a synchronous non-local game and $t \in \{loc,q,qa,qc,C^*,hered,A^*\}$, then we let $E_t(\cG)$ be the statement that $\cG$ has a winning $t$-strategy. Then
\[ E_{loc}(\cG) \implies E_q(\cG) \implies E_{qa}(\cG) \implies E_{qc}(\cG) \implies E_{C^*}(\cG) \implies E_{hered}(\cG) \implies E_{A^*}(\cG).\]
Most of these implications cannot be reversed in general. The fact that $E_q(\cG)$ does not imply $E_{loc}(\cG)$ in general follows from many different synchronous games, such as the graph coloring game (see \cite{MR16}, for example), the magic square game (see \cite{Mer90}) or the graph isomorphism game \cite{AMRSSV19}. The first examples of synchronous games with winning $qa$-strategies but no winning $q$-strategies were found in \cite{KPS18}, based off of the ground-breaking separation of $C_q(n,k)$ and $C_{qa}(n,k)$ for large $n,k$ \cite{Slo19}. In the resolution of the weak Tsirelson problem, a synchronous game $\cG$ was exhibited where $E_{qc}(\cG)$ holds but $E_{qa}(\cG)$ fails \cite{JNVWY20}. We have been pointed by C. Paddock and W. Slofstra to an example of a synchronous game $\cG$ where $E_{C^*}(\cG)$ is true, but $E_{qc}(\cG)$ fails \cite{PS21}. For the homomorphism game $\text{Hom}(K_5,K_4)$ from the complete graph $K_5$ to the complete graph $K_4$, it is known that $\cA(\text{Hom}(K_5,K_4)) \neq \{0\}$, but $\text{Hom}(K_5,K_4)$ has no winning hereditary strategy. It is not known whether $E_{hered}(\cG)$ implies $E_{C^*}(\cG)$ for all synchronous games $\cG$. For a discussion on the relation between $E_{hered}(\cG)$ and $E_{C^*}(\cG)$, we refer the reader to \cite{HMPS19}.

One way to relate synchronous games and their winning strategies is via $*$-equivalence. This notion of equivalence has appeared in several places \cite{KPS18,G21}. Synchronous non-local games $\cG_1$ and $\cG_2$ are said to be $*$\textbf{-equivalent} if there are unital $*$-homomorphisms $\pi:\cA(\cG_1) \to \cA(\cG_2)$ and $\rho:\cA(\cG_2) \to \cA(\cG_1)$. We note that $\pi$ and $\rho$ need not be mutual inverses or even be injective. In particular, if $\cA(\cG_1)$ and $\cA(\cG_2)$ are $*$-isomorphic as algebras, then $\cG_1$ and $\cG_2$ are $*$-equivalent. For any $t \in \{loc,q,qa,qc,C^*,hered,A^*\}$, if the games $\cG_1$ and $\cG_2$ are $*$-equivalent, then $\cG_1$ has a winning $t$-strategy if and only if $\cG_2$ has a winning $t$-strategy. (This follows readily by composing the appropriate homomorphisms.)

For games whose game algebras are $*$-isomorphic via a nice enough isomorphism, the constraints on winning strategies in the models $loc,q,qa,qc$ are stronger than what is implied by $*$-equivalence.

\begin{theorem}
\label{theorem: isomorphism gives affine homeomorphisms}
Let $\cG_1=(I,O,\lambda)$ and $\cG_2=(\widetilde{I},\widetilde{O},\widetilde{\lambda})$ be synchronous non-local games, with game algebras $\cA(\cG_1)$ and $\cA(\cG_2)$, generated by elements $e_{a,x}, \, a \in O, \, x \in I$ and $f_{c,z}, \, c \in \widetilde{O}, \, z \in \widetilde{I}$, respectively. Suppose that $\pi:\cA(\cG_1) \to \cA(\cG_2)$ is a $*$-isomorphism such that
\begin{enumerate}
\item
$\displaystyle \pi(e_{a,x})=\sum_{\substack{c \in \widetilde{O}, \\ z \in \widetilde{I}}} \alpha_{a,x,c,z} f_{c,z}, \, \, \alpha_{a,x,c,z} \in \bC$; and
\item
$\displaystyle \pi^{-1}(f_{c,z})=\sum_{\substack{a \in O, \\ x \in I}} \beta_{c,z,a,x} e_{a,x}, \, \, \beta_{c,z,a,x} \in \bC$.
\end{enumerate}
Let $t \in \{loc,q,qa,qc\}$. Then $C_t(\lambda)$ and $C_t(\widetilde{\lambda})$ are affinely homeomorphic.
\end{theorem}

\begin{proof}
We work in the case when $t=qc$; the other cases are similar. We define $\Phi:C_{qc}(\widetilde{\lambda}) \to C_{qc}(\lambda)$ by 
\[ (\Phi(r))(a,b|x,y)=\sum_{\substack{c,d \in \widetilde{O}, \\ z,w \in \widetilde{I}}} \alpha_{a,x,c,z}\alpha_{b,y,d,w} r(c,d|z,w), \, \, \, r \in C_{qc}(\widetilde{\lambda}),\]
and $\Psi:C_{qc}(\lambda) \to C_{qc}(\widetilde{\lambda})$ by
\[ (\Psi(p))(c,d|z,w)=\sum_{\substack{a,b \in O, \\ x,y \in I}} \beta_{c,z,a,x}\beta_{d,w,b,y} p(a,b|x,y), \, \, \, p \in C_{qc}(\lambda).\] 
To see that $\Phi(r)$ belongs to $C_{qc}(\lambda)$, we write $r(c,d|z,w)=\tau(F_{c,z}F_{d,w})$ for PVM's $\{F_{c,z}\}_{c \in \widetilde{O}}$ in a tracial $C^*$-algebra $(\cA,\tau)$, where $\tau$ is faithful. The definition of $r$ uniquely determines $\tau$ on all words of length two with letters in the generators $\{F_{c,z}\}$. Define $q_2:\cA(\cG_2) \to \cA$ by $q_2(f_{c,z})=F_{c,z}$, which is a unital $*$-homomorphism by universality of $\cA(\cG_2)$. Then define
\[ E_{a,x}=q_2 \circ \pi(e_{a,x})=\sum_{\substack{c \in \widetilde{O}, \\ z \in \widetilde{I}}} \alpha_{a,x,c,z} F_{c,z}, \, \, a \in O, \, x \in I.\]
Since $q_2$ and $\pi$ are unital $*$-homomorphisms, $\{E_{a,x}\}_{a \in O}$ is a PVM for each $x \in I$. Moreover,
\[ (\Phi(r))(a,b|x,y)=\tau(E_{a,x}E_{b,y}), \, \, a,b \in O, \, \, x,y \in I.\]
Note that $p$ is uniquely determined by $r$ and does not depend on the representation of $r$ in a tracial $C^*$-algebra $\cA$. In this way, one obtains a mapping $\Phi:C_t(\widetilde{\lambda}) \to C_t(\lambda)$. It is easy to see that $\Phi$ preserves convex combinations and is continuous. Similarly, $\Psi$ maps $C_{qc}(\lambda)$ into $C_{qc}(\widetilde{\lambda})$ and is a continuous affine mapping.

To see that $\Phi$ and $\Psi$ are mutual inverses, let $r(c,d|z,w)=\tau(F_{c,z}F_{d,w}) \in C_{qc}(\widetilde{\lambda})$ and $p(a,b|x,y)=\tau(E_{a,x}E_{b,y}) \in C_{qc}(\lambda)$ be as above, with $\Phi(r)=p$. Then in the tracial $C^*$-algebra $(\cA,\tau)$, we have
\[ \Psi(p)(c,d|z,w)=\tau(q_1 \circ \pi^{-1}(f_{c,z}) q_1 \circ \pi^{-1}(f_{d,w}))=\tau(q_1 \circ \pi^{-1}(f_{c,z}f_{d,w})),\]
where $q_1:\cA(\cG_1) \to \cA$ is the quotient mapping $q_1(e_{a,x})=E_{a,x}$ for all $a,x$. By construction, $q_1=q_2 \circ \pi$, so
\[ \Psi(p)(c,d|z,w)=\tau(q_2 \circ \pi \circ \pi^{-1}(f_{c,z}f_{d,w}))=\tau(q_2(f_{c,z}f_{d,w}))=\tau(F_{c,z}F_{d,w})=r(c,d|z,w).\]
Thus, $\Psi \circ \Phi(r)=r$ for every $r \in C_{qc}(\widetilde{\lambda})$. A similar argument shows that $\Phi \circ \Psi=\id_{C_{qc}(\lambda)}$, completing the proof in the case when $t=qc$.

For the remaining cases, the proof is almost identical. The only difference is that the correlation will be represented on an abelian $C^*$-algebra with a faithful (tracial) state in the case of $t=loc$. In the $t=q$ case, the correlation is represented on a finite-dimensional $C^*$-algebra, and in the $t=qa$ case, the correlation is represented in an ultrapower $\cR^{\cU}$ of the hyperfinite $II_1$-factor $\cR$ by a free ultrafilter $\cU$ over $\bN$, with the canonical trace.
\end{proof}

\section{Bisynchronous Games}

A non-local game $\cG=(I,O,\lambda)$ is said to be \textbf{bisynchronous} if:
\begin{align*}
\lambda(a,b,x,x)&=0, \, \text{ for all } a \neq b \text{ (synchronous rule)} \\
\lambda(a,a,x,y)&=0, \, \text{ for all } x \neq y \text{ (bisynchronous rule)}.
\end{align*}
Every synchronous game with $n$ inputs and $k$ outputs is $*$-equivalent to a bisynchronous game with $n$ inputs and $nk$ outputs. The main idea of this construction is that players respond to questions with allowable answers \textit{and} the question that they received \cite{PR21}. However, the case of bisynchronous games with identical question and answer sets is much more interesting: any winning hereditary strategy corresponds to a quantum permutation over a (hereditary) unital $*$-algebra $\cA$, subject to orthogonality conditions imposed by the rule function of the game. We recall that a quantum permutation over $\cA$ is a collection of elements $p_{ij} \in \cA$, where $1 \leq i,j \leq m$, such that each $p_{ij}$ is self-adjoint; $\displaystyle \sum_j p_{ij}=\sum_i p_{ij}=1$ for all $i,j$; and $p_{ij}p_{ik}=\delta_{jk} p_{ij}$ and $p_{ij}p_{\ell j}=\delta_{i\ell} p_{ij}$, where $\delta_{jk}$ denotes the usual Dirac delta. In this way, bisynchronous games on $m$ inputs and $m$ outputs are closely related to the quantum permutation group $S_m^+$ \cite{PR21}.

In the setting of a bisynchronous game with $m$ questions and $m$ answers and a winning hereditary strategy given by projections $\{E_{a,x}\}_{a,x=1}^m$, the quantum permutation obtained is the matrix $U=(E_{a,x})_{a,x=1}^m$. Indeed, each entry is a self-adjoint idempotent, and row and column sums are PVMs in the hereditary case \cite{PR21}. Some caution is required: in the game algebra $\cA(\cG)$ with generators $e_{a,x}$, while $\sum_{a=1}^m e_{a,x}=1$ and $e_{a,x}e_{b,x}=0$ for $a \neq b$ and $e_{a,x}e_{a,y}=0$ for $x \neq y$, there is no guarantee that $\sum_{x=1}^m e_{a,x}=1$; that is, the column sums of the matrix $(e_{a,x})_{a,x=1}^m$ are $1$, but the row sums might not be $1$ in general. For the bisynchronous games that we obtain, however, this extra condition will always hold in the game algebra.

The main idea of obtaining a bisynchronous game (with equal question and answer sets) from a synchronous game is as follows. Suppose that $p_0,...,p_{k-1}$ are self-adjoint idempotents in a unital $*$-algebra $\cA$, with $\sum_{a=0}^{k-1} p_a=1$ and $p_ap_b=0$ for $a \neq b$. Set $q_{a,b}=p_{a-b}$, where index calculations are done mod $k$. Then we obtain a quantum permutation $Q=(q_{a,b}) \in M_k(\cA)$, in the sense that every entry of $Q$ is a self-adjoint idempotent, every row sum and column sum of $Q$ is $1$, and $q_{a,b}q_{a,c}=0$ for $b \neq c$ and $q_{b,a}q_{c,a}=0$ for $b \neq c$. While this approach shows the algebraic motivation for the bisynchronous game, it is also helpful to see this bisynchronous game from the perspective of the rule function of the original synchronous game.

Given a synchronous game $\cG=(I,O,\lambda)$ with $O=\{0,1,...,k-1\}$, we briefly explain the idea of the bisynchronous game that we will use. The question and answer sets are both $O \times I$. In a round of the bisynchronous game, each player receives a question and an answer from the game $\cG$, represented by an element of $O \times I$. Suppose that Alice receives the pair $(a,x)$ and Bob receives the pair $(b,y)$. They respond with pairs $(i,v)$ and $(j,w)$, respectively, from $O \times I$. The players win if $x=v$ and $y=w$, and if the answers $a-i \mod{k}$ and $b-j \mod{k}$ are allowed for the question pair $(x,y)$ in the game $\cG$.

Next, we provide the formal definition of this bisynchronous game.

\begin{definition}
\label{definition: bisynchronous game associated to G}
Let $\cG=(I,O,\lambda)$ be a synchronous game with $|I|=n$ and $|O|=k$, and write $I=\{1,...,n\}$ and $O=\{0,...,k-1\}$. We define a bisynchronous game $\widetilde{\cG}=(O \times I,O \times I,\widetilde{\lambda})$ associated with $\cG$ by the rule function $\widetilde{\lambda}$ defined by
\begin{align}
\widetilde{\lambda}((i,v),(j,w),(a,x),(b,y))&=0 \text{ if } v \neq x \text{ or } w \neq y \label{no off-diagonal blocks for OxI}\\
\widetilde{\lambda}((i,v),(i,v),(a,x),(b,y))&=0 \text{ if } (a,x) \neq (b,y) \label{bisynchronous rule for OxI}\\
\widetilde{\lambda}((i,v),(j,w),(a,x),(a,x))&=0 \text{ if } (i,v) \neq (j,w) \label{synchronous rule for OxI}\\
\widetilde{\lambda}((i,x),(j,y),(a,x),(b,y))&=\lambda(a-i,b-j,x,y) \label{preserve lambda on OxI}
\end{align}
and $\widetilde{\lambda}((i,v),(j,w),(a,x),(b,y))=1$ otherwise. (In rule (\ref{preserve lambda on OxI}), index calculations are done modulo $k$.) 
\end{definition}

We first prove that certain relations involving the canonical generators of the game algebra $\cA(\widetilde{\cG})$ must hold.

\begin{lemma}
\label{lemma: relations in bisynchronous game}
Let $\cG=(I,O,\lambda)$ be a synchronous game with $I=\{1,...,n\}$ and $O=\{0,...,k-1\}$, and let $\widetilde{\cG}=(O\times I,O \times I,\widetilde{\lambda})$ be the bisynchronous game associated to $\cG$ as in Definition \ref{definition: bisynchronous game associated to G}. Let $f_{(i,v),(a,x)}$ be the canonical generators for $\cA(\widetilde{\cG})$, where $(i,v),(a,x) \in O \times I$. Then $f_{(i,v),(a,x)}=0$ if $v \neq x$, and $f_{(i,x),(a,x)}=f_{(j,x),(b,x)}$ whenever $a-i \equiv b-j \mod{k}$.
\end{lemma}

\begin{proof}
If $v \neq x$, then using Rule (\ref{no off-diagonal blocks for OxI}) with $x=y$, we obtain $f_{(i,v),(a,x)}=f_{(i,v),(a,x)}^2=0$. Next, we show that $f_{(i,x),(a,x)}=f_{(j,x),(b,x)}$ if $a-i=b-j$ mod $k$. Indeed, if $c,\ell \in O$ and $a-i \neq c-\ell$ mod $k$, then $\widetilde{\lambda}((i,x),(\ell,x),(a,x),(c,x))=\lambda(a-i,c-\ell,x,x)=0$ because $\cG$ is synchronous, so that $f_{(i,x),(a,x)}f_{(\ell,x),(c,x)}=0$ in $\cA(\widetilde{\cG})$.
Then, choosing $a,b \in O$ and using the fact that $\{f_{(\ell,y),(b,x)}\}_{(\ell,y) \in O \times I}$ is a PVM in $\cA(\widetilde{\cG})$, we have
\[ f_{(i,x),(a,x)}=\sum_{(\ell,y) \in O \times I} f_{(i,x),(a,x)}f_{(\ell,y),(b,x)}=\sum_{\ell \in O} f_{(i,x),(a,x)}f_{(\ell,x),(b,x)}=f_{(i,x),(a,x)}f_{(j,x),(b,x)},\]
where $j \in O$ is the unique element satisfying $a-i=b-j$ mod $k$. By multiplying $f_{(j,x),(b,x)}$ on the left by $1=\sum_{\ell,y} f_{(\ell,y),(a,x)}$ instead, we obtain $f_{(i,x),(a,x)}f_{(j,x),(b,x)}=f_{(j,x),(b,x)}$. Therefore, $f_{(i,x),(a,x)}=f_{(j,x),(b,x)}$ whenever $a-i \equiv b-j \mod k$.
\end{proof}

\begin{theorem}
\label{theorem: bisynchronous equivalence}
Let $\cG$ be a synchronous game on $n$ inputs and $k$ outputs. Then there is a bisynchronous game $\widetilde{\cG}$ on $nk$ inputs and $nk$ outputs such that the game algebras $\cA(\cG)$ and $\cA(\widetilde{\cG})$ are $*$-isomorphic.
\end{theorem}

\begin{proof}
We write $\cG=(I,O,\lambda)$, where $|I|=n$, $|O|=k$ and $\lambda:O\times O\times I \times I \to \{0,1\}$ is the rule function. To make calculations easier, we write $O=\{0,1,...,k-1\}$. Consider the game $\widetilde{\cG}=(O \times I,O \times I,\widetilde{\lambda})$ associated with $\cG$ as in Definition \ref{definition: bisynchronous game associated to G}. Consider the game algebra $\cA(\cG)$, which is generated by self-adjoint idempotents $e_{a,x}$, for $0 \leq a \leq k-1$ and $1 \leq x \leq n$, satisfying $\sum_{a=0}^{k-1} e_{a,x}=1$ and $e_{a,x}e_{b,y}=0$ whenever $\lambda(a,b,x,y)=0$. In particular, we have $e_{a,x}e_{b,x}=0$ if $a \neq b$, since $\cG$ is a synchronous game. For $(i,v),(a,x) \in O \times I$, define 
\[p_{(i,v),(a,x)}=\begin{cases} 0 & v \neq x \\ e_{a-i,x} & v=x,\end{cases}\]
where we interpret $a-i$ as $a-i$ mod $k$. Then each $p_{(i,v),(a,x)}$ is a self-adjoint idempotent, while
\[ \sum_{(i,v) \in O \times I} p_{(i,v),(a,x)}=\sum_{i \in O} e_{a-i,x}=\sum_{b \in O} e_{b,x}=1.\]

If $v \neq x$, then since $p_{(i,v),(a,x)}=0$, we also have $p_{(i,v),(a,x)}p_{(j,w),(b,y)}=0$. Similarly, if $w \neq y$, then $p_{(i,v),(a,x)}p_{(j,w),(b,y)}=0$, so Rule (\ref{no off-diagonal blocks for OxI}) holds. If $(a,x) \neq (b,y)$, then
\[ p_{(i,v),(a,x)}p_{(i,v),(b,y)}=\delta_{vx}\delta_{vy}e_{a-i,x}e_{b-i,y}.\]
If $x \neq y$, then the above quantity is $0$. In the case when $x=y$, then since $(a,x) \neq (b,y)$, we have $a \neq b$, and we obtain $e_{a-i,x}e_{b-i,x}=0$. Thus, Rule (\ref{bisynchronous rule for OxI}) holds.

Next, suppose that $(i,v) \neq (j,w)$. Then we have
\[ p_{(i,v),(a,x)}p_{(j,w),(a,x)}=\delta_{vx}\delta_{wx} e_{a-i,x}e_{a-j,x},\]
which is $0$ when $v \neq w$. If $v=w$, then since $(i,v) \neq (j,w)$, we have $i \neq j$, so the above quantity is $e_{a-i,x}e_{a-j,x}=0$. Thus, rule (\ref{synchronous rule for OxI}) is satisfied. Lastly, if $a,b,i,j \in O$ and $x,y \in I$ are such that $\lambda(a-i,b-j,x,y)=0$, then we observe that
\[ p_{(i,x),(a,x)}p_{(j,y),(b,y)}=e_{a-i,x}e_{b-j,y}=0,\]
by definition of $\cA(\cG)$, so that Rule (\ref{preserve lambda on OxI}) holds. By the universal property of $\cA(\widetilde{\cG})$, there is a unital $*$-homomorphism $\pi:\cA(\widetilde{\cG}) \to \cA(\cG)$ satisfying $\pi(f_{(i,v),(a,x)})=p_{(i,v),(a,x)}$ for all $i,v,a,x$, where $f_{(i,v),(a,x)}$ are the canonical generators of $\cA(\widetilde{\cG})$. 

Conversely, consider the game algebra $\cA(\widetilde{\cG})$, with canonical generators $f_{(i,v),(a,x)}$, where $(i,v),(a,x) \in O \times I$. By Lemma \ref{lemma: relations in bisynchronous game}, $f_{(i,v),(a,x)}=0$ if $v \neq x$, and $f_{(i,x),(a,x)}=f_{(j,x),(b,x)}$ whenever $a-i \equiv b-j \mod{k}$. In the case when $b=0$ and $a,i$ are fixed, we obtain $f_{(i,x),(a,x)}=f_{(i-a,x),(0,x)}$, where we interpret $i-a$ as being done modulo $k$.

It follows that, for any $(i,x) \in O \times I$,
\[ \sum_{(a,y) \in O \times I} f_{(i,x),(a,y)}=\sum_{a \in O} f_{(i,x),(a,x)}=\sum_{a \in O} f_{(i-a,x),(0,x)}=\sum_{b \in O} f_{(b,x),(0,x)}=\sum_{(b,y) \in O \times I} f_{(b,y),(0,x)}=1.\]

We define
\[ q_{a,x}=f_{(0,x),(a,x)}, \, \forall a \in O, \, x \in I.\]
Then $q_{a,x}$ is a self-adjoint idempotent, and
\[ \sum_{a=0}^{k-1} q_{a,x}=\sum_{a=0}^{k-1} f_{(0,x),(a,x)}=\sum_{(a,y) \in O \times I} f_{(0,x),(a,y)}=1,\]
using the fact that, $\sum_{(a,y) \in O \times I} f_{(i,x),(a,y)}=1$ for any $(i,x) \in O \times I$. If $\lambda(a,b,x,y)=0$, then
\[ q_{a,x}q_{b,y}=f_{(0,x),(a,x)}f_{(0,y),(b,y)}=0,\]
by Rule (\ref{preserve lambda on OxI}). Therefore, the mapping $\rho:e_{a,x} \mapsto q_{a,x}$ extends to a unital $*$-homomorphism from $\cA(\cG)$ to $\cA(\widetilde{\cG})$.

Lastly, we show that $\pi$ and $\rho$ are mutual inverses. For $(a,x) \in O \times I$, we have $\pi \circ \rho(e_{a,x})=\pi(f_{(0,x),(a,x)})=e_{a,x}$ for all $(a,x) \in O \times I$, so that $\pi \circ \rho=\id_{\cA(\cG)}$. If $v \neq x$, then $f_{(i,v),(a,x)}=0$, so that $\rho \circ \pi(f_{(i,v),(a,x)})=0=f_{(i,v),(a,x)}$. In the other case, we have
\[ \rho \circ \pi(f_{(i,x),(a,x)})=\rho(e_{a-i,x})=f_{(0,x),(a-i,x)}=f_{(i,x),(a,x)}.\]
Therefore, $\rho \circ \pi$ is the identity on the generators of $\cA(\widetilde{\cG})$, implying that $\rho \circ \pi=\id_{\cA(\widetilde{\cG})}$. We conclude that $\pi$ and $\rho$ are mutual inverses, so that $\cA(\cG)$ and $\cA(\widetilde{\cG})$ are $*$-isomorphic.
\end{proof}

\begin{corollary}\label{corollary: bisynchronous game with bad algebra}
There exists a bisynchronous game $\cG$ with $20$ inputs and $20$ outputs such that $\cA(\cG) \neq (0)$, but $\cG$ has no winning $qc$-strategy.
\end{corollary}

\begin{proof}
The homomorphism game $\text{Hom}(K_5,K_4)$ is a game with $5$ inputs and $4$ outputs that has $\cA(\text{Hom}(K_5,K_4)) \neq (0)$, but $\text{Hom}(K_5,K_4)$ has no winning $qc$-strategy \cite{HMPS19}. By the previous theorem, we obtain a bisynchronous game $\cG$ with $(5)(4)=20$ inputs and $20$ outputs with $\cA(\cG) \simeq \cA(\text{Hom}(K_5,K_4))$. Thus, $\cA(\cG) \neq (0)$, but $\cG$ has no winning $qc$-strategy.
\end{proof}

V. Paulsen and M. Rahaman asked in \cite{PR21} whether all bisynchronous games $\cG$ (with identical question and answer sets) satisfying $\cA(\cG)\neq(0)$ had a winning $qc$-strategy. Evidence in this direction was given by the graph isomorphism game \cite{BCEHPSW20} and the isometry game for finite metric spaces \cite{E20}, both of which have this property. However, Corollary \ref{corollary: bisynchronous game with bad algebra} answers this problem negatively. In fact, since $\cA(\text{Hom}(K_5,K_4))$ does not even have a winning hereditary strategy, the game in Corollary \ref{corollary: bisynchronous game with bad algebra} does not even have any representations into any unital hereditary $*$-algebra.

More oddities can be explored. For example, C. Paddock and W. Slofstra \cite{PS21} have constructed a synchronous non-local game that has no winning $qc$-strategy, but whose game algebra has a non-zero $C^*$-representation. Using Theorem \ref{theorem: bisynchronous equivalence}, it follows that there is a bisynchronous game $\cG=(I,O,\lambda)$ with $|I|=|O|$ such that $\cA(\cG)$ has a non-zero $C^*$-representation, but no tracial state (otherwise $\cG$ would have a winning $qc$-strategy). These examples are in contrast to the graph isomorphism game, where the existence of a trace, the existence of a $C^*$-representation, and the existence of the game algebra, are all equivalent \cite{BCEHPSW20}.

On the other hand, if one only deals with correlations rather than non-local games, then the synchronous correlations naturally occur as a closed face of a larger bisynchronous correlation set. This is similar to how $C_t^s(n,k)$ arises as a closed face of $C_t^s(nk,2)$ (see \cite{Ru20} and also \cite{Ha19}).

For a synchronous game $\cG=(I,O,\lambda)$, we recall that $C_t(\lambda)$ denotes the set of possible winning strategies for $\cG$, where $t \in \{loc,q,qa,qc\}$. Note that $C_t(\lambda)$ is a face in $C_t^s(n,k)$, where $n=|I|$ and $k=|O|$. Moreover, if $\cG$ is bisynchronous, then $C_t(\lambda)$ is a face in $C_t^{bs}(n,k)$.

\begin{theorem}
\label{theorem: include synchronous in the bisynchronous}
Let $\cG=(I,O,\lambda)$ be a synchronous game with $I=\{1,...,n\}$ and $O=\{0,...,k-1\}$, and let $\widetilde{\cG}=(O\times I,O\times I,\widetilde{\lambda})$ be the bisynchronous game associated with $\cG$ from Definition \ref{definition: bisynchronous game associated to G}. For each $t \in \{loc,q,qa,qc\}$, the sets $C_t(\lambda)$ and $C_t(\widetilde{\lambda})$ are affinely homeomorphic.

In particular, $C_t^s(n,k)$ is affinely homeomorphic to a face of $C_t^{bs}(nk,nk)$ for each $t \in \{loc,q,qa,qc\}$. Moreover, if $t \neq q$, then this face is closed.
\end{theorem}

\begin{proof}
We observe that the isomorphism $\cA(\cG) \to \cA(\widetilde{\cG})$ in Theorem \ref{theorem: bisynchronous equivalence} sends generators to sums of generators, as does its inverse. As a result of Theorem \ref{theorem: isomorphism gives affine homeomorphisms}, the sets $C_t(\lambda)$ and $C_t(\widetilde{\lambda})$ are affinely homeomorphic. The claim about $C_t^s(n,k)$ being affinely homeomorphic to a face of $C_t^{bs}(nk,nk)$ follows from these theorems and the fact that  $C_t^s(n,k)$ is the set of winning $t$-strategies for the trivial synchronous game given by
\[ \lambda(a,b|x,y)=\begin{cases} 1 & x \neq y \\ \delta_{ab} & x=y. \end{cases}\]
Lastly, this face will be closed if $t \neq q$, since the set $C_t^s(n,k)$ is closed for $t \in \{loc,qa,qc\}$.
\end{proof}

\begin{corollary}\label{corollary: synchronous separations to bisynchronous separations}
If $t_1,t_2 \in \{loc,q,qa,qc\}$ and $C_{t_1}^s(n,k) \neq C_{t_2}^s(n,k)$, then $C_{t_1}^{bs}(nk,nk) \neq C_{t_2}^{bs}(nk,nk)$.
\end{corollary}

\begin{corollary}
\label{corollary: bisynchronous separation}
Let $\cG$ be a synchronous game with $|I|=n$ and $|O|=k$, with a winning $qc$-strategy, but no winning $qa$-strategy (see \cite[Theorem~12.11]{JNVWY20}). Then there is a bisynchronous game $\widetilde{\cG}$ with $nk$ inputs and $nk$ outputs with a winning $qc$-strategy, but no winning $qa$-strategy. In particular, $C_{qa}^{bs}(nk,nk) \neq C_{qc}^{bs}(nk,nk)$.
\end{corollary}

The first part of the next corollary is certainly known due to the resolution of Connes' embedding problem \cite{JNVWY20}; the second part appears to be new and relies on Corollary \ref{corollary: bisynchronous separation}. In what follows, we denote by $C(S_N^+)$ the universal unital $C^*$-algebra generated by elements $u_{ij}$, $1 \leq i,j \leq n$, such that $u_{ij}^2=u_{ij}^*=u_{ij}$ for all $i,j$, and $U=(u_{ij})$ is unitary. (Equivalently, each $u_{ij}$ is an orthogonal projection and $\sum_{i=1}^n u_{ij}=\sum_{j=1}^n u_{ij}=1$ for all $i,j$.) This algebra arises as the full $C^*$-algebra of the quantum permutation group on $N$ letters.

\begin{corollary}
If $C_{qa}^s(n,k) \neq C_{qc}^s(n,k)$ and $C(S_N^+)$ represents the full ($C^*$-algebraic) quantum permutation group on $N$ letters, then $C(S_{nk}^+) \otimes_{\min} C(S_{nk}^+) \neq C(S_{nk}^+) \otimes_{\max} C(S_{nk}^+)$. Moreover, this non-isomorphism is witnessed at the level of tensor products of the canonical generators of $C(S_{nk}^+)$.
\end{corollary}

\begin{proof}
By Corollary \ref{corollary: bisynchronous separation}, we can choose an element $(p(a,b|x,y)) \in C_{qc}^{bs}(nk,nk) \setminus C_{qa}^{bs}(nk,nk)$. Then there are projections $E_{a,x}$ in a unital $C^*$-algebra $\cA$ and a faithful trace $\tau$ on $\cA$ such that $p(a,b|x,y)=\tau(E_{a,x}E_{b,y})$; moreover, $E=(E_{a,x}) \in M_{nk}(\cA)$ is a quantum permutation. Let $U=(u_{a,x}) \in M_{2nk}(C(S_{nk}^+))$ be the fundamental unitary. By universality of $C(S_{nk}^+)$, there is a unital $*$-homomorphism $\pi:C(S_{nk}^+) \to \cA$ satisfying $\pi(u_{a,x})=E_{a,x}$ for each $a,x$. Then $\rho=\tau \circ \pi$ is a trace on $C(S_{nk}^+)$ such that $\rho(u_{a,x}u_{b,y})=p(a,b|x,y)$. Then the linear functional $s:C(S_{nk}^+) \otimes_{\max} C(S_{nk}^+)^{op} \to \bC$ given by $s(x \otimes y^{op})=\rho(xy)$ is a state with respect to the maximal tensor product (see, for example, \cite{Oz13}).

If $s$ factors through the minimal tensor product, then a result of Kirchberg (\cite[Proposition 3.2]{Ki94}, \cite[Theorem 6.2.7]{BO08}) implies that $\rho$ is an amenable trace. By \cite[Theorem 3.6]{KPS18}, the correlation $(p(a,b|x,y))$ belongs to $C_{qa}^s(nk,nk)$. As $p$ satisfies the bisynchronous condition, we obtain $p \in C_{qa}^{bs}(nk,nk)$, contradicting Corollary \ref{corollary: bisynchronous separation}. Hence, $C(S_{nk}^+) \otimes_{\min} C(S_{nk}^+) \neq C(S_{nk}^+) \otimes_{\max} C(S_{nk}^+)$. 
\end{proof}

We would like to briefly point out how Corollary \ref{corollary: synchronous separations to bisynchronous separations} applies to other known separations of models. The smallest known separation between $C_q^s$ and $C_{qa}^s$ is $C_q^s(5,2) \neq C_{qa}^s(5,2)$ \cite{DPP19}. It directly follows, by Corollary \ref{corollary: synchronous separations to bisynchronous separations} that $C_q^{bs}(10,10) \neq C_{qa}^{bs}(10,10)$. As the smallest known game separation between the $q$ and $qa$ models for bisynchronous games is a game with $24$ inputs and $24$ outputs (arising from graphs that are $qa$-isomorphic but not $q$-isomorphic \cite{KPS18}), this appears to be a significant reduction. We note, though, that the separation $C_q^s(5,2) \neq C_{qa}^s(5,2)$ is not directly due to a synchronous game, and so our separation $C_q^{bs}(10,10) \neq C_{qa}^{bs}(10,10)$ is also not directly due to a synchronous game. Moreover, it is not even known if $\overline{C_q^{bs}(m,m)}=C_{qa}^{bs}(m,m)$ in general \cite{PR21}, which contrasts with the synchronous setting, where $\overline{C_q^s(n,k)}=C_{qa}^s(n,k)$ for all $n,k$ \cite{KPS18}.

On the other hand, the fact that $C_{loc}^s(3,2) \neq C_q^s(3,2)$ \cite{LR17} implies that $C_{loc}^{bs}(6,6) \neq C_q^{bs}(6,6)$. It should be noted that Corollary \ref{corollary: synchronous separations to bisynchronous separations} above does not yield any information when $n=k=2$. Indeed, B. Lackey and N. Rodrigues proved that, for $n \geq 2$, the set $C_{loc}^s(2,n)$ is precisely the set of all \textit{non-signalling} synchronous correlations that are symmetric; that is, for which $p(a,b|x,y)=p(b,a|y,x)$ holds for all $a,b,x,y$. It is well-known that every quantum commuting synchronous correlation is symmetric, since such correlations are given by tracial states \cite{PSSTW16}. Thus, $C_{loc}^s(2,n)=C_{qc}^s(2,n)$. We refer the reader to \cite[Appendix~A]{LR17} for the details.

We would like to point out that the maps from $C_t^s(n,k)$ into $C_t^{bs}(nk,nk)$ above do not work as nicely in the general framework of non-signalling correlations. We recall that the set of \textbf{non-signalling correlations} $C_{ns}(n,k)$ is simply the set of tuples $(p(a,b|x,y))$, $1 \leq a,b \leq k$ and $ 1\leq x,y \leq n$, satisfying
\begin{itemize}
\item
$p(a,b|x,y) \geq 0$ and $\displaystyle \sum_{a,b=1}^k p(a,b|x,y)=1$ for all $x,y$;
\item
$\displaystyle \sum_{b=1}^k p(a,b|x,y)$ is independent of the choice of $y$;
\item
$\displaystyle \sum_{a=1}^k p(a,b|x,y)$ is independent of the choice of $x$.
\end{itemize}
The last two conditions imply that there are well-defined marginal distributions given by
\[ p(a|x)=\sum_{b=1}^k p(a,b|x,y) \text{ and } p(b|y)=\sum_{a=1}^k p(a,b|x,y).\]
With $\lambda$ as the trivial synchronous rule function and $\widetilde{\lambda}$ as the associated bisynchronous rule function, it's not hard to see that, if $t \in \{loc,q,qa,qc\}$, then elements of $C_t(\widetilde{\lambda})$ carry the beneficial property that $q((i,v),(j,w)|(a,x),(b,y))=\delta_{vx}\delta_{wy} p(a-i,b-j|x,y)$ for a unique $p \in C_t(\lambda)$. However, $C_{ns}(\widetilde{\lambda})$ contains elements that fail this extra property in general. Indeed, consider the case when $|I|=|O|=2$ and $\lambda$ is the trivial synchronous rule function given by $\lambda(a,b,x,x)=\delta_{ab}$ and $\lambda(a,b,x,y)=1$ if $x \neq y$. The map $\Phi:C_{ns}(\lambda) \to C_{ns}(\widetilde{\lambda})$ from Theorems \ref{theorem: include synchronous in the bisynchronous} and \ref{theorem: isomorphism gives affine homeomorphisms} would be given by
\[ \Phi(p)((i,v),(j,w)|(a,x),(b,y))=\delta_{vx}\delta_{wy}p(a-i,b-j|x,y).\]
If $\Phi:C_{ns}(\lambda) \to C_{ns}(\widetilde{\lambda})$ were an affine homeomorphism, then the range of $\Phi$ would be a closed face in $C_{ns}^{bs}(4,4)$.

To show that the range of $\Phi$ is not a face in $C_{ns}^{bs}(4,4)$, we define three elements of $C_{ns}(\widetilde{\lambda})$:
\begin{align*}
p((i,v),(j,w)|(a,x),(b,y))&=\begin{cases} \frac{1}{2} \delta_{vx}\delta_{wy}\delta_{a-i,b-j} \text{ if } x=y \\ \frac{1}{4}\delta_{vx}\delta_{wy} \text{ if } x \neq y \end{cases} \\
q((i,v),(j,w)|(a,x),(b,y))&=\frac{1}{2}\delta_{vx}\delta_{wy}\delta_{a-i,b-j} \\
r((i,v),(j,w)|(a,x),(b,y))&=\begin{cases} \frac{1}{2}\delta_{vx}\delta_{wy}\delta_{a-i,b-j} \text{ if } x=y \\ \frac{1}{2}\delta_{vx}\delta_{wy}(1-\delta_{a-i,b-j}) \text{ if } x \neq y \end{cases}
\end{align*}
Each of these elements belong to $C_{ns}(\widetilde{\lambda})$. We note that $p$ belongs to the range of $\Phi$, but $q$ and $r$ do not. Moreover, $p=\frac{1}{2}(q+r)$, so the range of $\Phi$ is not a face in $C_{ns}(\widetilde{\lambda})$, and hence not a face in $C_{ns}^{bs}(4,4)$. Thus, $\Phi$ is not an affine homeomorphism.

\begin{remark}
It is possible to have two bisynchronous games with $*$-isomorphic game algebras, but with one game having no winning non-signalling strategies and the other having winning non-signalling strategies. Indeed, let $\cG_1=\text{Iso}(G,H)$ be a graph isomorphism game between two graphs $G$ and $H$ that are non-signalling isomorphic but not $qc$-isomorphic (see, for example, \cite{AMRSSV19}). Let $\cG_2=\text{Iso}(K,L)$ be a graph isomorphism game between two graphs that are not non-signalling isomorphic. Then $\cA(\cG_1)=(0)$; indeed, if it were non-zero, then $\cG_1$ would have a winning $qc$-strategy \cite{BCEHPSW20}, which it does not. Since non-signalling isomorphism is more general than $qc$-isomorphism, we must have $\cA(\cG_2)=(0)$ as well. However, the set of non-signalling winning strategies is non-empty, convex and closed for $\cG_1$, but the set of non-signalling winning strategies for $\cG_2$ is empty.
\end{remark}

\section{Synchronous games with three outputs}

In this section, we exhibit a $*$-isomorphism between the algebra of a synchronous game $\cG=(I,O,\lambda)$ with $|I|=n$ and $|O|=k>3$ and the algebra of a synchronous game $\widetilde{\cG}=(\widetilde{I},\widetilde{O},\widetilde{\lambda})$ with $|\widetilde{I}|=n(k-2)$ and $|\widetilde{O}|=3$. Due to work of T. Fritz, it is already known that every synchronous game algebra is $*$-isomorphic to a synchronous game algebra with at most $3$ outputs per question. The result comes as a consequence of the fact that every synchronous game algebra is $*$-isomorphic to a hypergraph $*$-algebra, and that every hypergraph $*$-algebra is $*$-isomorphic to a synchronous game algebra with at most $3$ outputs \cite{Fr20}. Our isomorphism here, while no longer immediately related to hypergraph algebras, is more optimal in the sense that the question set size of the final game algebra is smaller. We also show that the isomorphism induces affine homeomorphisms at the level of the sets of winning $t$-strategies, for $t \in \{loc,q,qa,qc\}$.

For winning strategies in a synchronous game, the key data that one must track are that the projections $E_{a,x}$ satisfy $E_{a,x}E_{b,y}=0$ whenever $\lambda(a,b,x,y)=0$, and that $\sum_{a=1}^k E_{a,x}=1$ for all $x$. The condition $\lambda(a,b,x,y)=0$ can easily be preserved in a game algebra with an inflated input set, but a smaller output set. Enforcing the condition that $\sum_{a=1}^k E_{a,x}=1$, however, requires more care when the desired output set is small. With three outputs, one has enough information in the game algebra to keep track of sums of projections, by using orthogonality conditions.

The content of the construction in this section relies heavily on a few lemmas about projections in a unital $*$-algebra. While these facts are all well-known, we include proofs for the reader's convenience.

\begin{lemma}\label{lemma: sum of two projections}
Let $p,q,r$ be projections in a unital $*$-algebra $\cA$. Then $p+q=r$ if and only if the relations
\begin{align}
pq&=0, \label{orthogonal}\\
p(1-r)&=0, \label{p below r}\\
q(1-r)&=0, \label{q below r}\\
(1-(p+q))r&=0 \label{p+q=r}
\end{align}
all hold.
\end{lemma}

\begin{proof}
If $p+q=r$, then since $r^2=r$, we have
\[ p+q=(p+q)^2=p+q+pq+qp.\]
It follows that $pq+qp=0$. Multiplying on the right by $p$ gives $qp=-pqp$. As $-pqp$ is self-adjoint, taking adjoints gives $qp=pq$. As $pq+qp=0$, we get $2pq=0$, yielding $pq=0$. Hence, equation (\ref{orthogonal}) holds. Multiplying the equation $p+q=r$ by $(1-r)$ yields $p(1-r)+q(1-r)=0$. Multiplying on the left by $p$ (respectively, $q$), we obtain $p(1-r)=0$ (respectively, $q(1-r)=0$), which are equations (\ref{p below r}) and (\ref{q below r}). Equation (\ref{p+q=r}) is immediately obtained by multiplying the equation $p+q=r$ on the left by $1-(p+q)$, since $p+q$ is a projection.

Conversely, let $p,q,r$ be projections satisfying the relations given. Using equation (\ref{orthogonal}), $(p+q)^2=p^2+pq+qp+q^2=p+q$, since $pq=0$, $p=p^*$ and $q=q^*$. Combining (\ref{p below r}) and (\ref{q below r}), we obtain $(p+q)(1-r)=0$. Thus, $p+q=(p+q)r$. Using (\ref{p+q=r}), we also get $r=(p+q)r$, yielding $p+q=r$.
\end{proof}

As a special case of Lemma \ref{lemma: sum of two projections}, we obtain the following simple lemma.

\begin{lemma}
\label{lemma: two projections are equal}
Let $\cA$ be a unital $*$-algebra, and let $p,r \in \cA$ be self-adjoint idempotents. Then $p=r$ if and only if $p(1-r)=0$ and $r(1-p)=0$.
\end{lemma}

\begin{proof}
This is an immediate application of Lemma \ref{lemma: sum of two projections} with $q=0$.
\end{proof}

The next lemma is well-known, but we include the proof for completeness.

\begin{lemma}
\label{lemma: PVMs with 3 outputs}
Let $\cA$ be a unital $*$-algebra, and let $p,q,r \in \cA$ be self-adjoint idempotents satisfying $p+q+r=1$. Then $pq=pr=qr=0$.
\end{lemma}

\begin{proof}
We will show that $pq=0$; similar arguments show that $pr=qr=0$. Since $r^2=r=r^*$, it follows that $(1-r)^2=1-r=(1-r)^*$; thus, $p+q=1-r$ is a self-adjoint idempotent. Lemma \ref{lemma: sum of two projections} yields $pq=0$. 
\end{proof}

\begin{lemma}
\label{lemma: second projection lemma}
Let $\{q_1,q_2,q_3\}$ and $\{r_1,r_2,r_3\}$ be two PVMs in a unital $*$-algebra $\cA$. Then $q_1+q_2=r_1$ if and only if $r_1q_3=0$ and $q_ir_j=0$ for $i=1,2$ and $j=2,3$.
\end{lemma}

\begin{proof}
First, assume that $q_1+q_2=r_1$. Then by Lemma \ref{lemma: sum of two projections}, $q_1(1-r_1)=q_2(1-r_1)=0$ and $(1-(q_1+q_2))r_1=0$. As $1-r_1=r_1+r_2$ and $1-(q_1+q_2)=q_3$, we obtain $r_1q_3=0$ and $q_1(r_2+r_3)=q_2(r_2+r_3)=0$. Expanding the first one gives
\[ q_1r_2+q_1r_3=0.\]
Since $r_1,r_2,r_3$ are self-adjoint idempotents summing to $1$, by Lemma \ref{lemma: PVMs with 3 outputs}, we have $r_2r_3=0$. Thus, multiplying the equation on the right by $r_2$, we obtain $q_1r_2=0$. By multiplying by $r_3$ instead of $r_2$, we obtain $q_1r_3=0$. Similarly, we obtain $q_2r_2=q_2r_3=0$.

Conversely, suppose that $r_1q_3=0$ and $q_ir_j=0$ for $i=1,2$ and $j=2,3$. We note that $q_1q_2=0$ since $\{q_1,q_2,q_3\}$ is a PVM on three outputs. Next, we have $q_1(1-r_1)=q_1(r_2+r_3)=q_1r_2+q_1r_3=0$ by assumption. Similarly, $q_2(1-r_1)=0$. Lastly, $(1-(q_1+q_2))r_1=q_3r_1=(r_1q_3)^*=0$. By Lemma \ref{lemma: sum of two projections}, $q_1+q_2=r_1$.
\end{proof}

Next, we show that we may always replace the rule function $\lambda$ with a symmetric rule function. By a \textbf{symmetric rule function}, we mean a rule function $\lambda: O \times O \times I \times I \to \{0,1\}$ that satisfies $\lambda(a,b,x,y)=\lambda(b,a,y,x)$ for all $a,b,x,y$. For a synchronous non-local game $\cG=(I,O,\lambda)$, the \textbf{symmetrized} game $\cG_{sym}$ is defined with rule function \[\lambda_{sym}(a,b,x,y)=\lambda(a,b,x,y)\lambda(b,a,y,x).\]
(See \cite{HMPS19}.) While this extra restriction of the rule function would make a difference, for example, in the set of winning non-signalling strategies, no such difference is seen at the level of the game algebra.

\begin{proposition}
\label{proposition: symmetric isomorphism}
Let $\cG=(I,O,\lambda)$ be a synchronous non-local game. Then $\cA(\cG)$ is  $*$-isomorphic to $\cA(\cG_{sym})$.
\end{proposition}

\begin{proof}
As $\cG$ and $\cG_{sym}$ have the same input set $I$ and the same output set $O$, we may write $e_{a,x}$ for the canonical generators of $\cA(\cG)$ and $f_{a,x}$ for the canonical generators of $\cG_{sym}$, where $x \in I$ and $a \in O$. If $\lambda(a,b,x,y)=0$, then $\lambda_{sym}(a,b,x,y)=0$ so that $f_{a,x}f_{b,y}=0$. By definition, we have $\sum_{a=1}^k f_{a,x}=1_{\cA(\mathcal{G}_{sym})}$, so the canonical map $e_{a,x} \mapsto f_{a,x}$ extends to a unital $*$-homomorphism $\pi:\cA(\cG)\to \cA(\cG_{sym})$.

Conversely, if $\lambda_{sym}(a,b,x,y)=0$, then either $\lambda(a,b,x,y)=0$ or $\lambda(b,a,y,x)=0$. In the first case, we have $e_{a,x}e_{b,y}=0$ in $\cA(\mathcal{G})$. In the latter case, we have $e_{b,y}e_{a,x}=0$; after taking adjoints, we obtain $e_{a,x}e_{b,y}=0$. Since $\sum_{a=1}^k e_{a,x}=1_{\cA(\cG)}$ for all $x$, the map $f_{a,x} \mapsto e_{a,x}$ extends to a unital $*$-homomorphism $\rho:\cA(\cG_{sym}) \to \cA(\cG)$. Evidently $\pi$ and $\rho$ are mutual inverses, so the game algebras are $*$-isomorphic.
\end{proof}

Next, we construct a synchronous game with three outputs associated to a synchronous game $\mathcal{G}$.

\begin{definition}
Let $\cG=(I,O,\lambda)$ be a symmetric synchronous game on $n$ inputs and $k$ outputs with $k>3$. We define the symmetric three-output game  $\widetilde{\cG}=(\widetilde{I},\widetilde{O},\widetilde{\lambda})$ associated to $\cG$ with $|\widetilde{I}|=n(k-2)$ and $|\widetilde{O}|=3$ as follows. For simplicity, we write $I=\{1,...,n\}$ and $O=\{1,...,k\}$. We define
\[ \widetilde{I}=\{1,...,k-2\} \times \{1,...,n\},\]
and set $\widetilde{O}=\{1,2,3\}$. Define $\mu:\widetilde{O} \times \widetilde{O} \times \widetilde{I} \times \widetilde{I} \to \{0,1\}$ by
\begin{align}
\mu(1,1,(1,x),(1,y))&=\lambda(1,1,x,y) \label{first lambda preserver} \\
\mu(1,2,(1,x),(b,y))&=\lambda(1,b+1,x,y) &1 \leq b \leq k-2 \\
\mu(1,3,(1,x),(k-2,y))&=\lambda(1,k,x,y) \\
\mu(2,2,(a,x),(b,y))&=\lambda(a+1,b+1,x,y) &1 \leq a,b \leq k-2 \\ 
\mu(2,3,(a,x),(k-2,y))&=\lambda(a+1,k,x,y) &1 \leq a \leq k-2 \label{final lambda preserver} \\ 
\mu(3,3,(k-2,x),(k-2,y))&=\lambda(k,k,x,y) \label{the actual final lambda preserver} \\
\mu(1,3,(a+1,x),(a,x))&=0 &1 \leq a \leq k-3 \label{r_a below p_a+p_{a+1}}\\
\mu(i,j,(a,x),(a+1,x))&=0 &i=1,2, \, j=2,3, \, 1 \leq a \leq k-3, \label{p_a+p_{a+1} below r_a}
\end{align}
and $\mu(i,j,(a,x),(b,y))=1$ otherwise, and set $\widetilde{\lambda}=\mu_{sym}$. 
\end{definition}

Throughout this section, we will write $f_{i,(a,x)}$ for the generators of the game algebra for $\widetilde{\cG}$, where $i=1,2,3$, $1 \leq a \leq k-2$ and $1 \leq x \leq n$. We define $p_{1,x}=f_{1,(1,x)}$, $p_{a,x}=f_{2,(a-1,x)}$ for $2 \leq a \leq k-1$, and $p_{k,x}=f_{3,(k-2,x)}$. We also set $r_{a,x}=f_{1,(a+1,x)}$ for all $1 \leq a \leq k-3$.

\begin{lemma}\label{lemma: p_{a,x} is a PVM}
For each $x$ and for each $1 \leq a \leq k-2$, we have $f_{1,(a,x)}+f_{2,(a,x)}=f_{1,(a+1,x)}$. In particular, $\sum_{a=1}^k p_{a,x}=1_{\mathcal{A}(\widetilde{\cG})}$.
\end{lemma}

\begin{proof}
We first note that $f_{1,(a+1,x)}f_{3,(a,x)}=0$ by Rule (\ref{r_a below p_a+p_{a+1}}). Similarly, using Rule (\ref{p_a+p_{a+1} below r_a}), we have
\[f_{i,(a,x)}f_{j,(a+1,x)}=0, \, i=1,2, \, j=2,3, \, \, 1 \leq a \leq k-3.\]
By Lemma \ref{lemma: second projection lemma}, it follows that $f_{1,(a+1,x)}=f_{1,(a,x)}+f_{2,(a,x)}$.

For the last statement, since $p_{1,x}=f_{1,(1,x)}$ and $p_{a,x}=f_{2,(a-1,x)}$ for $2 \leq a \leq k-1$ and $p_{k,x}=f_{3,(k-2,x)}$, we see that
\[ p_{1,x}+p_{2,x}=f_{1,(1,x)}+f_{2,(1,x)}=f_{1,(2,x)}=r_{1,x}.\]
Similarly, for each $3 \leq a \leq k-2$, \[p_{a,x}+r_{a-2,x}=f_{2,(a-1,x)}+f_{1,(a-1,x)}=f_{1,(a,x)}=r_{a-1,x}.\] It follows that $r_{k-3}=\sum_{a=1}^{k-2} p_{a,x}$. Using the fact that $\{r_{k-3},p_{k-1},p_k\}$ is a PVM in $\cA(\widetilde{\cG})$, it follows that $\sum_{a=1}^k p_{a,x}=1_{\cA(\widetilde{\cG})}$.
\end{proof}

\begin{theorem}
\label{theorem: equivalence to 3 output game}
If $\cG=(I,O,\lambda)$ is a synchronous game on $n$ inputs and $k$ outputs with $k>3$, then there is a symmetric synchronous game $\widetilde{\cG}$ on $n(k-2)$ inputs and $3$ outputs such that $\cA(\cG)$ and $\cA(\widetilde{\cG})$ are $*$-isomorphic.
\end{theorem}

\begin{proof}
We let $e_{a,x}$ be the canonical generators of the game algebra $\cA(\cG)$, and we let $f_{i,(a,x)}$ be the canonical generators of $\cA(\widetilde{\cG})$. As before, we define $p_{1,x}=f_{1,(1,x)}$, $p_{a,x}=f_{2,(a-1,x)}$ for $2 \leq a \leq k-1$ and $p_{k,x}=f_{3,(k-2,x)}$. Notice that rules (\ref{first lambda preserver})--(\ref{the actual final lambda preserver}) show that $p_{a,x}p_{b,y}=0$ whenever $\lambda(a,b,x,y)=0$. By Lemma \ref{lemma: p_{a,x} is a PVM}, we have $\sum_{a=1}^k p_{a,x}=1$ for all $x$. By the universal property of $\cA(\cG)$, there is a unital $*$-homomorphism $\pi:\cA(\cG) \to \cA(\widetilde{\cG})$ such that $\pi(e_{a,x})=p_{a,x}$ for all $1 \leq x \leq n$ and $1 \leq a \leq k$.

In $\cA(\cG)$, for convenience, we define $s_{1,x}=e_{1,x}+e_{2,x}$, and recursively define $s_{a,x}=e_{a+1,x}+s_{a-1,x}$ for all $2 \leq a \leq k-3$. We define projections $q_{i,(a,x)}$ for $1 \leq a \leq k-2$, $1 \leq x \leq n$ and $i=1,2,3$ as follows. For each $x$, we set
\begin{align}
q_{1,(1,x)}&=e_{1,x} \\
q_{2,(a,x)}&=e_{a+1,x} & 1 \leq a \leq k-2 \\
q_{1,(a,x)}&=s_{a-1,x} & 2 \leq a \leq k-2 \\
q_{3,(a,x)}&=1-s_{a,x} & 1 \leq a \leq k-3 \\
q_{3,(k-2,x)}&=e_{k,x}.
\end{align}
By construction of each $s_{a-1,x}$ and the fact that each set $\{e_{a,x}\}_{a=1}^k$ is a set of mutually orthogonal projections summing to $1$, it is evident that $\{q_{1,(a,x)},q_{2,(a,x)},q_{3,(a,x)}\}$ is a set of mutually orthogonal projections in $\cA(\cG)$ summing to $1$, for each $1 \leq a \leq k-2$ and $1 \leq x \leq n$. Now, we show that these projections satisfy the rules of $\widetilde{\cG}$.

First, $q_{1,(1,x)}q_{1,(1,y)}=e_{1,x}e_{1,y}$, while $q_{1,(1,x)}q_{2,(b,y)}=e_{1,x}e_{b+1,y}$, $q_{1,(1,x)}q_{3,(k-2,y)}=e_{1,x}e_{k,y}$, $q_{2,(a,x)}q_{2,(b,y)}=e_{a+1,x}e_{b+1,y}$, and $q_{2,(a,x)}q_{3,(k-2,y)}=e_{a+1,x}e_{k,y}$ and $q_{3,(k-2,x)}q_{3,(k-2,y)}=e_{k,x}e_{k,y}$. It immediately follows that rules (\ref{first lambda preserver})--(\ref{the actual final lambda preserver}) are satisfied by the projections $q_{i,(a,x)}$, since $e_{a,x}e_{b,y}=0$ whenever $\lambda(a,b,x,y)=0$. Next, we observe that, by definition of $s_{a,x}$, we have $q_{1,(2,x)}=s_{1,x}=e_{1,x}+e_{2,x}=q_{1,(1,x)}+q_{2,(1,x)}$ and, for $2 \leq a \leq k-2$, $q_{1,(a+1,x)}=s_{a,x}=s_{a-1,x}+e_{a+1,x}=q_{1,(a,x)}+q_{2,(a,x)}$. Thus, for all $1 \leq a \leq k-2$,
\[ q_{1,(a+1,x)}=q_{1,(a,x)}+q_{2,(a,x)}.\]
By Lemma \ref{lemma: second projection lemma}, it follows that $q_{1,(a+1,x)}q_{3,(a,x)}=0$ and $q_{i,(a,x)}q_{j,(a+1,x)}=0$ for $i=1,2$ and $j=2,3$ and $1 \leq a \leq k-3$. Thus, Rules (\ref{r_a below p_a+p_{a+1}}) and (\ref{p_a+p_{a+1} below r_a}) are satisfied.

By the universal property of $\cA(\widetilde{\cG})$, there is a unital $*$-homomorphism $\rho:\cA(\widetilde{\cG}) \to \cA(\cG)$ such that $\rho(f_{i,(a,x)})=q_{i,(a,x)}$ for all $i \in \widetilde{O}$ and $(a,x) \in \widetilde{I}$. We observe that $\rho(\pi(e_{1,x}))=\rho(f_{1,(1,x)})=q_{1,(1,x)}=e_{1,x}$, while
\[ \rho(\pi(e_{a,x}))=\rho(p_{a,x})=\rho(f_{2,(a-1,x)})=q_{2,(a-1,x)}=e_{a,x}, \, 2 \leq a \leq k-1\]
and $\rho(\pi(e_{k,x}))=\rho(p_{k,x})=\rho(f_{3,(k-2,x)})=q_{3,(k-2,x)}=e_{k,x}$. Thus, $\rho \circ \pi=\id_{\cA(\cG)}$. Similarly, one can check that $\pi \circ \rho=\id_{\cA(\widetilde{\cG})}$, so $\cA(\cG)$ and $\cA(\widetilde{\cG})$ are $*$-isomorphic.
\end{proof}

\begin{remark}
By \cite[Theorem~12.11]{JNVWY20}, there exists a synchronous non-local game $\cG=(I,O,\lambda)$ with $|I|=n$ and $|O|=k$ for some very large $n,k$, that has a winning $qc$ strategy, but no winning $qa$ strategy. By Theorem \ref{theorem: equivalence to 3 output game}, there is a synchronous non-local game $\widetilde{\cG}$ with $n(k-2)$ questions and $3$ answers that has a winning $qc$ strategy, but no winning $qa$-strategy.
\end{remark}

Applying Theorem \ref{theorem: equivalence to 3 output game} to the $4$-coloring game for the complete graph $K_5$ yields the following corollary. 

\begin{corollary}
There exists a synchronous non-local game with $10$ inputs and $3$ outputs that has non-zero game algebra, but no winning $qc$-strategy.
\end{corollary}

\begin{proof}
We recall that the $4$-coloring game for $K_5$, denoted $\text{Hom}(K_5,K_4)$, has $5$ inputs and $4$ outputs, and has non-zero game algebra but no winning $qc$-strategy \cite{HMPS19}. By the construction in Theorem \ref{theorem: equivalence to 3 output game}, we obtain a synchronous game $\cG$ with $(5)(4-2)=10$ inputs and $3$ outputs with $\cA(\cG)\neq(0)$ but such that $\cG$ has no winning $qc$-strategy.
\end{proof}

We close by showing that, if one inflates the input set for a synchronous game $\cG$, then one can arrive at a game $\widetilde{\cG}$ with $3$ outputs whose game algebra relations are given by a certain set of projections being zero, and certain pairs of projections being equal, in addition to the projection-valued measure relations. Moreover, the new game algebra $\cA(\widetilde{\cG})$ will still be $*$-isomorphic to $\cA(\cG)$. In this way, we can remove orthogonality relations from the set of relations defining the game algebra; moreover, this transformation yields an affine homeomorphism at the level of winning strategies in the models $loc,q,qa,qc$.

\begin{theorem}
\label{theorem: zero set and relation set}
Let $\cG=(I,O,\lambda)$ be a synchronous game. Then $\cA(\cG)$ is $*$-isomorphic to a game algebra $\cA(\widetilde{\cG})$ that is the unital $*$-algebra generated by self-adjoint idempotent entries $e_{a,z}$, $z \in \widetilde{I}$, $a \in O$, modulo an ideal of the form $\cJ$ generated by elements of the form $e_{c,z}$ for all pairs $(c,z)$ belonging to a certain set $Z$, and elements of the form $e_{c,z}-e_{d,w}$ for all tuples $(c,z,d,w)$ belonging to a certain set $R$.
\end{theorem}

\begin{proof}
Let $\cG=(I,O,\lambda)$ be a synchronous game. By Theorem \ref{theorem: equivalence to 3 output game}, we may assume that $|O|=3$. We define $\widetilde{\cG}=(I \dot{\cup} \lambda^{-1}(\{0\}),O,\widetilde{\lambda})$, where $\dot{\cup}$ denotes a disjoint union. For simplicity, we denote elements of $\lambda^{-1}(\{0\})$ by $4$-tuples $(c,d,z,w)$. We define $\widetilde{\lambda}$ by
\begin{align}
\widetilde{\lambda}(a,b,x,x)&=0 \text{ if } x,y \in I \text{ and } a \neq b \\
\widetilde{\lambda}(i,j,(c,d,z,w),(c,d,z,w))&=0 \text{ if } i \neq j \\
\widetilde{\lambda}(a,i,x,(a,b,x,y))&=0 \text{ if } i=2,3 \label{e_{a,x} below f_{1,(a,b,x,y)}}\\
\widetilde{\lambda}(j,1,x,(a,b,x,y))&=0 \text{ if } j \neq a \label{1-e_{a,x} below 1-f_{1,(a,b,x,y)}}\\
\widetilde{\lambda}(b,i,y,(a,b,x,y))&=0 \text{ if } i=1,3 \label{e_{b,y} below f_{2,(a,b,x,y)}} \\
\widetilde{\lambda}(j,2,y,(a,b,x,y))&=0 \text{ if } j \neq b. \label{1-e_{b,y} below 1-f_{2,(a,b,x,y)}}
\end{align}
We write the generators of $\cA(\cG)$ as $e_{a,x}$ for $x \in I$ and $a \in O$. In $\cA(\widetilde{\cG})$, we write $f_{i,(a,b,x,y)}$ as the generators where $i \in \{1,2,3\}$ and $(a,b,x,y) \in \lambda^{-1}(\{0\})$, and $g_{a,x}$ as the generators for $x \in I$ and $a \in O$. Fix an element $(a,b,x,y)$ of $\lambda^{-1}(\{0\})$. Using Rule (\ref{e_{a,x} below f_{1,(a,b,x,y)}}), it follows that $g_{a,x}f_{i,(a,b,x,y)}=0$ for $i=2,3$, so that $g_{a,x}(1-f_{1,(a,b,x,y)})=0$. On the other hand, if $j \neq a$ in $\{1,2,3\}$, then by Rule (\ref{1-e_{a,x} below 1-f_{1,(a,b,x,y)}}), we have $g_{j,x}f_{1,(a,b,x,y)}=0$. Therefore, $(1-g_{a,x})f_{1,(a,b,x,y)}=0$. By Lemma \ref{lemma: two projections are equal}, $g_{a,x}=f_{1,(a,b,x,y)}$. Similarly, rules (\ref{e_{b,y} below f_{2,(a,b,x,y)}}) and (\ref{1-e_{b,y} below 1-f_{2,(a,b,x,y)}}) force $g_{b,y}=f_{2,(a,b,x,y)}$. It follows that $g_{a,x}$ and $g_{b,y}$, along with $f_{3,(a,b,x,y)}$, are three projections summing to $1$ in $\cA(\widetilde{\cG})$. By Lemma \ref{lemma: PVMs with 3 outputs}, $g_{a,x}g_{b,y}=0$. By definition of $\cA(\cG)$, the generators $g_{a,x}$ satisfy $\sum_{a=1}^3 g_{a,x}=1$. Therefore, there is a unital $*$-homomorphism $\pi:\cA(\cG) \to \cA(\widetilde{\cG})$ such that $\pi(e_{a,x})=g_{a,x}$ for all $a,x$.

For the converse direction, in $\cA(\cG)$ we define, for each $(a,b,x,y) \in \lambda^{-1}(\{0\})$, $h_{1,(a,b,x,y)}=e_{a,x}$, $h_{2,(a,b,x,y)}=e_{b,y}$ and $h_{3,(a,b,x,y)}=1-(e_{a,x}+e_{b,y})$. Since $e_{a,x}e_{b,y}=0$, each $h_{i,(a,b,x,y)}$ is a projection in $\cA(\cG)$, for $i=1,2,3$ and $(a,b,x,y) \in \lambda^{-1}(\{0\})$. It is immediate that Rules (\ref{e_{a,x} below f_{1,(a,b,x,y)}})--(\ref{1-e_{b,y} below 1-f_{2,(a,b,x,y)}}) are satisfied by the projections $h_{i,(a,b,x,y)}$. Therefore, we obtain a unital $*$-homomorphism $\rho:\cA(\widetilde{\cG}) \to \cA(\cG)$ satisfying $\rho(g_{a,x})=e_{a,x}$, $\rho(f_{1,(a,b,x,y)})=h_{1,(a,b,x,y)}=e_{a,x}$ and $\rho(f_{2,(a,b,x,y)})=h_{2,(a,b,x,y)}=e_{b,y}$. It is not hard to verify that $\rho \circ \pi$ and $\pi \circ \rho$ are the identity on the generators of the respective algebras, which shows that $\cA(\cG)$ and $\cA(\widetilde{\cG})$ are $*$-isomorphic.
\end{proof}

\begin{remark}
In Theorem \ref{theorem: zero set and relation set}, if the game $\cG$ is symmetric, then one does not need the whole set $\lambda^{-1}(\{0\})$, as $\lambda(a,b,x,y)=\lambda(b,a,y,x)$ for all $a,b,x,y$.
\end{remark}

In the game obtained in the proof of Theorem \ref{theorem: zero set and relation set}, winning strategies are given by an array of $3$-output PVMs $(f_{a,x})_{a=1,2,3}$ where the rules are summarized in a zero set $Z \subset \{1,2,3\} \times \{1,2,...,n\}$ and a relation set $R \subset (\{1,2,3\} \times \{1,2,...,n\})^2$, satisfying $f_{a,x}=f_{b,y}$ if $(a,b,x,y) \in R$, and $f_{a,x}=0$ if $(a,x) \in Z$.

Given $n \in \bN$, and sets $Z \subset \{1,2,3\} \times \{1,...,n\}$ and $R \subset (\{1,2,3\} \times \{1,...,n\})^2$, we will define, for $t \in \{loc,q,qa,qc\}$, the set $C_t(Z,R)$ as the set of all elements $p=(p(a,b|x,y)) \in C_t^s(n,3)$ satisfying the following:
\begin{itemize}
\item
$p(a,b|x,y)=0$ if $(a,x) \in Z$ or $(b,y) \in Z$;
\item
$p(a,b|x,y)=p(a',b|x',y)$ if $(a,a',x,x') \in R$; and
\item
$p(a,b|x,y)=p(a,b'|x,y')$ if $(b,b',y,y') \in R$.
\end{itemize}

\begin{corollary}
Let $\cG=(I,O,\lambda)$ be a synchronous non-local game. Then there exists an $n \in \bN$ and sets $Z \subset \{1,2,3\} \times \{1,...,n\}$ and $R \subset (\{1,2,3\} \times \{1,...,n\})^2$ such that $C_t(\lambda)$ is affinely homeomorphic to $C_t(Z,R)$.
\end{corollary}

\begin{proof}
The game algebra $\cA(\cG)$ is $*$-isomorphic to a game algebra $\cA(\widetilde{\cG})$ with $n$ inputs of the form in Theorem \ref{theorem: zero set and relation set} for a zero set $Z$ and a relation set $R$. These relations are easily encoded in the rule function $\mu:\{1,2,3\} \times \{1,2,3\} \times \widetilde{I} \times \widetilde{I} \to \{0,1\}$ given by
\begin{itemize}
\item 
$\mu(a,b,x,y)=0$ if $(a,x) \in Z$ or $(b,y) \in Z$; and
\item
If $(a,b,x,y) \in R$, then $\mu(a',b,x,y)=\mu(a,b',x,y)=0$ for all $a' \neq a$ and $b' \neq b$.
\end{itemize}
As the isomorphisms in Proposition \ref{proposition: symmetric isomorphism} and Theorems \ref{theorem: equivalence to 3 output game} and \ref{theorem: zero set and relation set} send generators to sums of generators, the rest of the proof follows by an application of Theorem \ref{theorem: isomorphism gives affine homeomorphisms}.
\end{proof}

\section*{Acknowledgements}

The author was supported in part by an NSERC postdoctoral fellowship. We would like to thank Michael Brannan, Matt Kennedy, Vern Paulsen, Connor Paddock, and William Slofstra for valuable discussions.


\begin{thebibliography}{99}
\bib{AMRSSV19}{article}{
	author={A. Atserias},
	author={L. Man\v{c}inska},
	author={D.E. Roberson},
	author={R. \v{S}\'{a}mal},
	author={S. Severini},
	author={A. Varvitsiotis},
	title={Quantum and non-signalling isomorphisms},
	journal={J. Comb. Theory, Series B},
	volume={136},
	year={2019},
	pages={289--328},
}
\bib{BCEHPSW20}{article}{
	author={M. Brannan},
	author={A. Chirvasitu},
	author={K. Eifler},
	author={S.J. Harris},
	author={V.I. Paulsen},
	author={X. Su},
	author={M. Wasilewski},
	title={Bigalois extensions and the graph isomorphism game},
	journal={Comm. Math. Phys.},
	volume={375},
	pages={1777--1809},
	year={2020},
	}
\bib{BO08}{book}{
author={N. Brown},
author={N. Ozawa},
title={${C}^{\ast}$-algebras and finite-dimensional approximations},
series={Graduate Studies in Mathematics}, publisher={American Mathematical Society Providence, RI}, 
year={2008},
}
\bib{Co76}{article}{
	author={A. Connes},
	title={Classification of injective factors. Cases $II_1$, $II_{\infty}$, $III_{\lambda}$, $\lambda \neq 1$},
	journal={Ann. Math. (2)},
	volume={104},
	number={1},
	year={1976},
	pages={73--115},
}
\bib{DPP19}{article}{
author={K. Dykema},
author={V.I. Paulsen},
author={J. Prakash},
title={Non-closure of the set of quantum correlations via graphs},
journal={Comm. Math. Phys.},
volume={365},
year={2019},
pages={1125--1142},
}
\bib{E20}{article}{
	author={K. Eifler},
	title={Non-local games and quantum symmetries of quantum metric spaces},
	journal={preprint (arXiv:2011.03867 [math.OA])},
	year={2020},
	}
\bib{Fr11}{article}{
	author={T. Fritz},
	title={Tsirelson's problem and Kirchberg's conjecture},
	journal={Rev. Math. Phys.},
	volume={24},
	number={5},
	eid={1250012},
	date={2012},
	pages={1250012, 67 pp.},
}
\bib{Fr20}{article}{
	author={T. Fritz},
	title={Curious properties of free hypergraph $C^{\ast}$-algebras},
	journal={J. Op. Theory},
	volume={83},
	number={2},
	year={2020},
	pages={423--445},
}
\bib{G21}{article}{
	author={A. Goldberg},
	title={Synchronous linear constraint system games},
	journal={J. Math. Phys.},
	volume={62},
	year={2021},
	pages={032201},
	}
\bib{Ha19}{article}{
	author={S.J. Harris},
	title={Unitary correlation sets and their applications},
	journal={PhD thesis},
	year={2019},
	}
\bib{HMPS19}{article}{
	author={J.W. Helton},
	author={K.P. Meyer},
	author={V.I. Paulsen},
	author={M. Satriano},
	title={Algebras, synchronous games, and chromatic numbers of graphs},
	journal={New York J. Math.},
	volume={25},
	year={2019},
	pages={328--361},
	}
\bib{JNVWY20}{article}{
	author={Z. Ji},
	author={A. Natarajan},
	author={T. Vidick},
	author={J. Wright},
	author={H. Yuen},
	title={MIP$^\ast$ = RE},
	journal={preprint (arXiv:2001.04383 [math-ph])},
	year={2020},
}
\bib{J+}{article}{
	author={M. Junge},
	author={M. Navascues},
	author={C. Palazuelos},
	author={D. Perez-Garcia},
	author={V.B. Scholz},
	author={R.F. Werner},
	title={Connes' embedding problem and Tsirelson's problem},
	journal={J. Math. Phys.},
	volume={52},
	year={2011},
	pages={012102, 12pp.},
	}
\bib{KPS18}{article}{
author={S.-J. Kim},
author={V.I. Paulsen},
author={C. Schafhauser},
title={A synchronous game for binary constraint systems},
journal={Journal of Mathematical Physics},
volume={59},
number={3},
year={2018},
pages={032201, 17pp.},
}

\bib{Ki94}{article}{
author={E. Kirchberg}, 
title={Discrete groups with {K}azhdan’s property {T} and factorization property are residually finite},
journal={Mathematische Annalen}, 
volume={299},
year={1994},
pages={35–-63},
}
\bib{LR17}{article}{
author={B. Lackey},
author={N. Rodrigues},
title={Nonlocal games, synchronous correlations, and Bell inequalities},
journal={preprint (arXiv:1707.06200 [quant-ph])},
year={2017},
}
\bib{MR16}{article}{
author={L. Man\v{c}inska},
author={D.E. Roberson},
title={Oddities of quantum colorings},
journal={Baltic J. Modern Computing},
volume={4},
number={4},
year={2016},
pages={846--859},
}
\bib{Mer90}{article}{
author={N.D. Mermin},
title={Simple unified form for the major no-hidden-variables theorems},
journal={Phys. Rev. Lett.,}
volume={65},
year={1990},
pages={3373–-3376},
}
\bib{Oz13}{article}{
author={N. Ozawa},
title={About the Connes embedding conjecture--algebraic approaches},
journal={Jpn. J. Math.},
volume={8},
year={2013},
pages={147--183},
}
\bib{PS21}{article}{
author={C. Paddock},
author={W. Slofstra},
title={private communication},
}
\bib{PR21}{article}{
author={V.I. Paulsen},
author={M. Rahaman},
title={Bisynchronous games and factorizable maps},
journal={Ann. Henri Poincar\'{e}},
volume={22},
pages={593--614},
year={2021},
}
\bib{PSSTW16}{article}{
author={V.I. Paulsen},
author={S. Severini},
author={D. Stahlke},
author={I.G. Todorov},
author={A. Winter},
title={Estimating quantum chromatic numbers},
journal={Journal of Functional Analysis},
volume={270},
number={6},
year={2016},
pages={2188--2222},
}
\bib{Ru20}{article}{
author={T.B. Russell},
title={Two-outcome synchronous correlation sets and Connes' embedding problem},
journal={Quantum Information and Computation}, 
volume={20},
number={5\&6},
year={2020},
pages={361--374},
}
\bib{Slo19}{article}{
author={W. Slofstra},
title={The set of quantum correlations is not closed},
journal={Forum of Mathematics, Pi},
volume={7},
publisher={Cambridge University Press},
year={2019},
pages={e1},
}
\bib{Ts93}{article}{
author={B.S. Tsirelson},
title={Some results and problems on quantum Bell-type inequalities},
	journal={Hadronic J. Supp.},
	fjournal={Hadronic Journal Supplement},
	volume={8},
	date={1993},
	number={4},
	pages={329--345},
}
\end{thebibliography}
\end{document}